\documentclass[conference]{IEEEtran}
\usepackage{graphicx}
\usepackage{amsmath}
\usepackage{lipsum}
\usepackage{amsthm}
\usepackage{amsfonts}
\usepackage{caption}
\usepackage{blkarray}
\usepackage{mathtools}
\usepackage{siunitx}
\usepackage{amssymb}
\usepackage{subcaption}
\usepackage{placeins}
\usepackage{xcolor}
\usepackage{lettrine}
\usepackage{multicol, blindtext}
\usepackage{bbm}
\usepackage{cite}
\usepackage{algorithm}
\usepackage{algpseudocode}
\usepackage{authblk}
\newtheorem{remark}{Remark}
\makeatletter
\newcommand\blfootnote[1]{%
  \begingroup
  \renewcommand\thefootnote{}\footnote{#1}%
  \addtocounter{footnote}{-1}%
  \endgroup
}

\newcommand{\Rmnum}[1]{\expandafter\@slowromancap\romannumeral #1@}

\newtheorem{theorem}{Theorem}

\newtheorem{proposition}{Proposition}

\newtheorem{definition}{Definition}
\newtheorem{lemma}{Lemma}

\makeatletter
\def\BState{\State\hskip-\ALG@thistlm}
\makeatother
\ifCLASSINFOpdf
\else
\fi

\begin{document}

\title{
When to pull data from sensors for minimum Age of incorrect Information}


\author{ Saad Kriouile$^\dagger$}\author{Mohamad Assaad$^*$}
\affil{$^\dagger$ Huawei Technologies, France \\$^*$Laboratoire des Signaux et Syst\`emes, CentraleSup\'elec, Universit\'e Paris-Saclay, France}

\maketitle
\blfootnote{This work has been done before 14 October 2022 in CentraleSupelec.}
\newcommand{\HRule}{\rule{\linewidth}{0.5mm}}

\begin{abstract}
The age of Information (AoI) has been introduced to capture the notion of freshness in real-time monitoring applications. However, this metric falls short in many scenarios, especially when quantifying the mismatch between the current and the estimated states.
To circumvent this issue, in this paper, we adopt the age of incorrect information metric (AoII) that considers the quantified mismatch between the source and the knowledge at the destination while tracking the impact of freshness. We consider for that a problem where a central entity pulls the information from remote sources that evolve according to a Markovian Process. It selects at each time slot which sources should send their updates.
As the scheduler does not know the actual state of the remote sources, it estimates at each time the value of AoII based on the Markovian sources' parameters.
Its goal is to keep the time average of the AoII function as small as possible.
For that purpose, we develop a scheduling scheme based on Whittle's index policy. 
To that extent, we use the Lagrangian Relaxation Approach and establish that the dual problem has an optimal threshold policy. Building on that, we compute the expressions of Whittle's indices. Finally, we provide some numerical results to highlight the performance of our derived policy compared to the classical AoI metric.

\end{abstract}

\section{Introduction}

The number of connected devices has witnessed considerable growth in the last decade due to the emergence of IoT. This remarkable proliferation of low-cost hardware has led to the emergence of real-time monitoring services.
In these systems, the monitor needs to know the status of one or multiple processes observed by remote sensors. 
To that extent, due to the channel constraint, the monitor selects only a subset of users from which it pulls the data that contains useful information to execute the convenient task.
The sensor is only responsible for sampling the process and transmitting the information content to the receiver on demand.

The main goal in these applications is to implement a scheduling scheme that keeps the monitor up to date by receiving fresh information from different sources.

This notion of freshness is captured by the Age of Information (AoI), which is introduced for the first time in \cite{kaul2012real}. More specifically, AoI can be viewed as the duration which separates the
generation of the last successfully received packet's time-stamp and the current time.

There are several works that consider this metric of AoI in different contexts and fields, and from different perspectives \cite{maatouk2020optimality,hsu2019scheduling,kadota2018scheduling,bedewy2019age,maatouk2020status,sun2018age,SunElif2017,Oguz2022}.
Nevertheless, even though this metric can quantify the information time lag at the monitor, it doesn't take into account the information content transmitted by the sensor. Several metrics have then been proposed in the literature to quantify the value and quality of the information \cite{zhong2018,Ayan2019,Escroc2019,Chiariotti2022,ElifSemantic}. 

In fact, AoI evolves regardless of the state of the remote source. For instance, in some scenarios, the state of the remote source is the same as the estimated state at the side of the monitor while AoI keeps growing. Whereas, no penalty should be incurred as the monitor is up to date. 

To meet the timeliness requirement while considering the content of the information sent, the authors in \cite{maatouk2020age} have designed a new metric dubbed Age of incorrect information. This metric grows only if the state of the remote sensor is different from the actual state and goes to zero otherwise. In \cite{Ayik2023}, an extension of this metric, called Age of Incorrect Information at Query (QAoII), is proposed. This metric considers that the information is only relevant at the times the receiver generates a query. 

The AoII has been analyzed in several papers and information sampling and scheduling schemes have been developed \cite{maatouk2020age,maatouk2020agenew,kam2020age,kriouile2021minimizing,chen2021minimizing,ChenAoII2021}. A Markovian source model is considered in \cite{maatouk2020age,maatouk2020agenew} under energy constraint. A threshold-based sampling policy has been derived and proved to be the optimal policy. A symmetric binary information source over a delay channel with feedback is considered in \cite{kam2020age}. The optimal information sampling policy is then derived by dynamic programming. Furthermore, in \cite{chen2021minimizing,ChenAoII2021}, the AoII metric is also considered and the authors assume that the scheduler has perfect knowledge about the source process at each time slot and restrict the analysis to one transmitter-receiver pair communication. Likewise, the authors in \cite{maatouk2020age,maatouk2020agenew,kam2020age} consider that the transmitter is responsible for observing at each time slot the state of the source in order to decide whether or not to send the packet. In the context of multiple sources/sensors, if each source decides on its own to transmit a packet, e.g. by random access techniques, collisions will occur which will reduce the system performance, especially in the context of high number of sensors (e.g. IoT scenario). A collision-free transmission can be obtained if  a central entity  decides whether a source must sample and transmit a packet or not. The central entity (e.g. monitor) cannot be aware of the status of the remote sources, and hence the AoII cannot be known perfectly by the scheduler. A prediction/estimation of the AoII must be performed, e.g. by  averaging over the different possible values of AoII. The predicted AoII will be then used in the scheduling decision by the monitor. This problem can be modeled as a Partially Observable Markov Decision Process (POMDP). To the best of our knowledge, the first work that studied such a framework in the context of AoI is   \cite{kriouile2021minimizing}, in which the  authors consider a multiple-transmitter-one-receiver scenario and aim to minimize the expected total average AoII by deriving the low complex and well-performing policy, called Whittle's index policy, which is optimal in the many-users regime.


However, in the aforementioned work, the mismatch between the source and monitor in the AoII metric is considered to be the indicator function. In other words, the distance between the different states of the Markovian source is not considered in  \cite{kriouile2021minimizing}, which makes it fall short in some real-life applications. For example, in temperature monitoring, the goal of the central entity could be to monitor the temperature variations of the system and quickly respond to these fluctuations. Thus, the freshness of information is not the only priority. We also need to monitor the temperature variations as high volatility or significant temperature variation is more harmful to the system than the smaller ones. This variation can be seen as the distance between the actual state of the remote source and the estimated state at the monitor's side: The further away the actual state is from the estimated state, the more we need to sanction the system.


For this reason, we consider in this work the AoII metric, in which the mismatch between the source state and the monitor knowledge is modeled as the distance between them. We tackle a realistic case in which a scheduler tracks the states of multiple remote sources and selects at each time a subset of them. These selected sources will transmit the information of interest while the others remain idle. The difficulty lies in the fact that the scheduler does not know the instantaneous state of the remote sources until it receives their updates. Our goal will be accordingly to minimize the total expected average AoII.  For that, we apply the well-known, well-performing, and low-complex Whittle's index policy referred by WIP \cite{Maialen18}. This policy has been widely adopted in the framework of the Age of Information (one can see \cite{tripathi2019whittle,kadota2018scheduling,sombabu2020whittle,maatouk2020optimality,SaadAoI2022} and the references therein). Contrary to these references,  we derive the expression of the estimated/predicted AoII at the monitor side and use it in the development of the WIP. Specifically, our contributions can be summarized as follows: i)  we consider a system where a central entity tracks the status of remote sensors. We formulate the AoII-based scheduling problem and show that it belongs to the family of Restles Multi-Armed Bandit (RMAB) problems. The challenge in this case is that the monitor/scheduler does not know the states of the sensors (i.e. the AoII exact value) and has then to predict its value in order to perform the scheduling, ii) as the optimal solution of this type of problem is known to be out of reach, we adopt the Lagrangian Relaxation approach that consists of relaxing and decomposing the problem in one-dimensional problems, iii) we establish the indexability of the one-dimensional problem by proving that its optimal solution is threshold-based policy, and iv) we derive the Whittle's index policy by leveraging the steady-state form of the one-dimensional problem under a given threshold policy.

\section{System Model}\label{sec:Syst_mod}
\subsection{Network description}\label{subsec:Net_descrip}
We consider in our paper $N$ sensors that generate and send status updates about the process of interest to a central entity over unreliable channels. Time is considered to be discrete and normalized to the time slot duration.
More precisely, when the monitor or the base station decides to schedule a given sensor $i$ at time $t$, this later samples its respective process, $X_i(t)$, and sends it to the monitor over an unreliable channel.
If the transmission is successful, the packet containing the information of interest $X_i(t)$ will be instantaneously delivered to the monitor. Then, the monitor keeps the information state $X_i(t)$ received by the sensor till the next successful transmission. That means that the information state at the side of the monitor denoted by $\hat{X}_i(t)$ is equal to $X_i(g_i(t))$ where $g_i(t)$ indicates the time-stamp of the last successfully received packet by the monitor.

As for the unreliable channel, we suppose that for user $i$, at each time slot $t$, the probability of having successful transmission is $\rho_i$, and $1-\rho_i$ otherwise. Consequently, the channel realizations are independent and identically distributed (i.i.d.) over time slots that we denote $c_i(t)$, i.e. $c_i(t)=1$ if the packet is successfully transmitted and $c_i(t)=0$ otherwise.

On the other hand, regarding the nature of the process of interest $X_i(t)$, we consider that for each user $i$, the process $X_i(t)$ evolves under Markov chain with infinite state space as represented in Figure \ref{fig:markov_source}. We define the probability of transitioning to the next state at the next time slot as $p_i$. Similarly, the probability of remaining at the same state is $1-p_i$. 
Furthermore, we consider that the distance between two consecutive states is $d_i$.

\begin{figure}[H]
\centering
\includegraphics[width=0.47\textwidth]{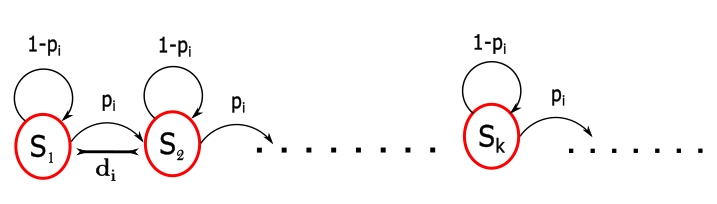}
\caption{Illustration of process $X_i(t)$}
\label{fig:markov_source}
\end{figure}
\subsection{Penalty function dynamics}
In this paper, we study the Age of incorrect information penalty function. We see how it is relevant and more realistic to consider the distance between the states of the source in order to have a good performance in some scenarios where the applications are sensitive to the gap between the estimated state and the current state of a given process. For that purpose, we start by reintroducing in the next section the standard metric, the age of information metric, to emphasize its shortcomings. Then we propose our adjusted AoII metric.

\subsubsection{Age of information penalty function}
The standard metric (AoI) that captures the freshness of information for user $i$ is: 
\begin{equation}
\delta_{AoI}(t)=t-g_i(t)
\end{equation}
where $g_i(t)$, as mentioned before,
is the time-stamp of the last successfully received packet by the monitor. This metric captures the lifetime of the last update at the monitor without taking into account the information state of the remote markovian source. Thereby, this makes it fall short in some applications. For instance, in some scenarios, the accuracy or performance of some applications relies heavily on information mismatch error-sensitive actions. Having said that, a higher penalty should be paid if the estimated source at the side of the monitor is far from the actual state of the remote source. In other words, the penalty function should be proportional to the distance between the estimated and the actual state. 


\subsubsection{Distance-based Age of incorrect information penalty function}\label{subsubsec:MAoII_pf} 
As was depicted in the Introduction, to capture the notion of the gap or the mismatch between the source and the monitor while satisfying the freshness requirement, one should adjust the AoII metric by integrating the distances between the states of the source. 
To that extent, we consider this following metric:
$$\sum_{u=V_i(t)}^t[X_i(u)-\hat{X}_i(t)]d_i$$
where $V_i(t)$ is the last time where $X_i(V_i(t))=\hat{X}_i(V_i(t))$.
\begin{remark}
Considering our system model described in \ref{sec:Syst_mod}, replacing $V_i(t)$ by $g_i(t)$ where $g_i(t)$ is the last successfully transmitted packet's time-stamp  gives us the same metric.
\end{remark}
To that extent, leveraging this remark above, we consider the following metric:
$$\sum_{u=g_i(t)}^t[X_i(u)-\hat{X}_i(t)]d_i$$
In this paper, we consider that the monitor that plays the scheduler's role knows only the state of the last successively received packet. Thus, the base station takes the expectation of AoII at each time slot. Accordingly, the explicit expression of the AoII metric in our case is:

\begin{equation}
\delta_{AoII}(t)=\mathbb{E}\sum_{u=g_i(t)}^t[X_i(u)-\hat{X}_i(t)]d_i
\end{equation}
In the sequel, we provide the closed-form expression of our metric in a Partially Observable Markov Decision Process Problem.

\subsection{Metrics evolution}\label{subsec:metrics_evolution}
In this section, we describe mathematically the evolution of our metric depending on the system parameters and the action taken.
We denote by $a_i(t)$ the action prescribed to user $i$ at time slot $t$ and by $s_i(t)$, the age of incorrect information penalty function at time slot $t$. 
According to the expression of AoII given in section \ref{subsubsec:MAoII_pf}, for $g_i(t) \leq u\leq t$, $X_i(u)-\hat{X}_i(u)$ is a random variable that we denote by $A_i(u)$ that satisfies:
\begin{lemma}\label{lem:random_variable}
\begin{align}
A_i(u)=\left\{
    \begin{array}{lll}
        u_{g_i}(0) & w.p & (p_i)^{u-g_i(t)}\\
        u_{g_i}(1) & w.p &  (p_i)^{u-g_i(t)-1}.\binom{u-g_i(t)}{1}(1-p_i)\\
        u_{g_i}(2) & w.p &  (p_i)^{u-g_i(t)-2}.\binom{u-g_i(t)}{2}(1-p_i)^2\\
        u_{g_i}(3) & \cdots  &         \cdots\\
        \vdots & &\\
        u_{g_i}(k) & w.p & (p_i)^{u-g_i(t)-k}.\binom{u-g_i(t)}{k}(1-p_i)^k\\
        \vdots & &\\              
        0& w.p & (1-p_i)^{u-g_i(t)} 
    \end{array}
\right.
\end{align}
\end{lemma}
where $\binom{n}{k}=\frac{n!}{(n-k)!k!}$ and $u_{g_i}(k)=u-g_i(t)-k$.
\begin{IEEEproof}
See appendix \ref{app:lem:random_variable}. 
\end{IEEEproof}
Therefore, the mean of AoII at slot $t$ equals to the mean of $\sum_{u=g_i(t)}^t d_i A_i(u)$, i.e.
\begin{lemma}\label{lem:CME_expression}
The mean of the AoII at slot $t$, denoted by $n_i(t)$ equals to:
$$n_i(t)=\sum_{u=g_i(t)}^t d_i\mathbb{E}[A_i(u)]=d_i p_i\frac{(t-g_i(t)+1)(t-g_i(t))}{2}$$
\end{lemma}
\begin{proof}
See Appendix \ref{app:lem:CME_expression}
\end{proof}

As $n_i(t)$ depends only $t-g_i(t)$ and $i$, then, we let $n_i(t)\overset{\Delta}{=}n_i(t-g_i(t))$. Therefore:
\begin{equation}
n_i(j)=d_i p_i\frac{(j+1)j}{2}
\end{equation}
To that extent, at time slot $t+1$, if the user $i$ is scheduled and the packet is successively transmitted, then $g_i(t+1)=t+1$. Accordingly, at time slot $t+1$, AoII equals to $n_i(t+1-g_i(t+1))=n_i(0)$.   
If the user $i$ is not scheduled or if the packet is not successively transmitted, then $g_i(t+1)=g_i(t)$. Therefore, AoII will transit to $n_i(t+1-g_i(t+1))=n_i(t-g_i(t)+1)$. To sum up, the evolution of AoII can be summarized as follows:     
\begin{equation}
 s_i(t+1)=\left\{
    \begin{array}{ll}
        n_i(0) & if d_i(t+1)=1, c_i(t+1)=1 \\
        n_i(j+1) & else
    \end{array}
\right.
\end{equation}
where $s_i(t)=n_i(j)$.

\section{Problem formulation}\label{sec:prob_form}
We let the vector $\boldsymbol s$ at time $t$ be $\boldsymbol{s}(t)=(s_{1}(t),\ldots,s_{N}(t))$ where $s_i(t)$ is the penalty function at the central entity of user $i$ with respect to AoII metric at time slot $t$. Our aim is to find a scheduling policy that allocates per each time slot, the available channels ($M$ channels) to a given subset of users ($M$ users, $M \leq N$) in a such way to minimize the total expected average AoII penalty function. A scheduling policy $\phi$ is defined as a sequence of actions $\phi=(\boldsymbol{a}^{\phi}(0),\boldsymbol{a}^{\phi}(1),\ldots)$ where $\boldsymbol{a}^{\phi}(t)=(a_1^{\phi}(t),a_2^{\phi}(t),\ldots,a_{N}^{\phi}(t))$ is a binary vector such that $a_i^{\phi}(t)=1$ if the user $i$ is scheduled at time $t$. 
Denoting by $\Phi$, the set of all causal scheduling policies, then
our scheduling problem can be formulated as follows:
\begin{equation}
\setlength{\belowdisplayskip}{0pt} \setlength{\belowdisplayshortskip}{0pt}
\setlength{\abovedisplayskip}{0pt} \setlength{\abovedisplayshortskip}{0pt} 
\begin{aligned}
& \underset{\phi\in \Phi}{\text{minimize}}
& & \lim_{T\to+\infty} \text{sup}\:\frac{1}{T}\mathbb{E}^{\phi\in \Phi}\Big(\sum_{t=0}^{T-1}\sum_{i=1}^{N}s_i^{\phi}(t)|\boldsymbol{s}(0)\Big)\\
& \text{subject to}
& & \sum_{i=1}^{N}a_{i}^{\phi}(t)\leq\alpha N \quad t=1,2,\ldots
\end{aligned}
\label{eq:original_problem}
\end{equation}
where $\alpha N=M$.
The problem in (\ref{eq:original_problem}) falls into Restless Bandit framework.
RMAB problems are known to be generally difficult to solve them as they are PSPACE-Hard \cite{papadimitriou1999complexity}. To circumvent this complexity, we propose to implement a low-complex and well-performing policy called Whittle's index policy (WIP) \cite{weber1990index}. In order to get the Whittle's index values, we need to adopt the Lagrangian relaxation. 
To that extent, we introduce in the next section the Lagrangian relaxation approach applied to our RBP problem. Then, we provide the mathematical analysis to get the Whittle's index policy (WIP).

\section{Lagrangian Relaxation and Whittle's Index}\label{sec:lag_relax_whi_ind}
\subsection{Relaxed problem}
In order to derive the Whittle's index scheduling policy, we adopt the Lagrangian relaxation technique. First, it consists of relaxing the constraint on the available resources by letting it be satisfied on average rather than in every time slot. More specifically, we define our Relaxed Problem (\textbf{RP}) as follows:
\begin{equation}\label{eq:relaxed_problem}
\setlength{\belowdisplayskip}{0pt} \setlength{\belowdisplayshortskip}{0pt}
\setlength{\abovedisplayskip}{0pt} \setlength{\abovedisplayshortskip}{0pt} 
\begin{aligned}
& \underset{\phi\in \Phi}{\text{minimize}}
& & \lim_{T\to+\infty} \text{sup}\:\frac{1}{T}\mathbb{E}^{\phi}\Big(\sum_{t=0}^{T-1}\sum_{i=1}^{N}s_i^{\phi}(t)|\boldsymbol{s}(0)\Big)\\
& \text{subject to}
& & \lim_{T\to+\infty}\text{sup}\frac{1}{T}\mathbb{E}^{\phi}\Big(\sum_{t=0}^{T-1}\sum_{i=1}^{N}a_{i}^{\phi}(t)\Big)\leq\alpha N
\end{aligned}
\end{equation}
\color{black}
The Lagrangian function $f(W,\phi)$ of the problem \eqref{eq:relaxed_problem} is defined as:
\begin{equation}
\lim_{T\to+\infty} \text{sup}\:\frac{1}{T}\mathbb{E}^{\phi}\Big(\sum_{t=0}^{T-1}\sum_{i=1}^{N}s_i^{\phi}(t)+Wa_{i}^{\phi}(t)|\boldsymbol{s}(0)\Big)-W\alpha N
\end{equation}
where $W \geq 0$ can be seen as a penalty for scheduling users. Thus, by following the Lagrangian approach, our next objective is to solve the following problem:

\begin{equation}\label{eq:dual_problem}
\\\\ \underset{\phi\in \Phi}{\text{min}} f(W,\phi)
\end{equation}

As the term $W\alpha N$ is independent of $\phi$, it can be eliminated from the analysis. Baring that in mind, we present the steps to obtain the Whittle's index policy:
\begin{enumerate}
\item We decompose the problem in (\ref{eq:dual_problem}) into $N$ one-dimensional problems that can be solved independently (this has been shown in \cite{kriouile:hal-03437753}). Accordingly, we drop the user's index for ease of notation, and we deal with the one-dimensional problem:
\begin{equation}\label{eq:individual_dual_problem}
\setlength{\belowdisplayskip}{0pt} \setlength{\belowdisplayshortskip}{0pt}
\setlength{\abovedisplayskip}{0pt} \setlength{\abovedisplayshortskip}{0pt} 
\begin{aligned}
& \underset{\phi\in \Phi}{\text{min}}
& & \lim_{T\to+\infty} \text{sup}\:\frac{1}{T}\mathbb{E}^{\phi}\Big(\sum_{t=0}^{T-1}s^{\phi}(t)+Wa^{\phi}(t)|s(0)\Big)
\end{aligned}
\end{equation}
\item We give the structural results on the optimal solution of the one-dimensional problem.
\item We establish the indexability property of Problem \ref{eq:individual_dual_problem}.
\item Under indexability condition, we derive a closed-form expression of the Whittle's index values
\item We define the proposed scheduling policy (WIP) for the original problem (\ref{eq:original_problem}).
\end{enumerate}
\subsection{Structural results}
The problem in (\ref{eq:individual_dual_problem}) can be viewed as an infinite horizon average cost Markov decision process that is defined as follows:
\begin{itemize}
\item \textbf{States}: The state of the MDP at time $t$ is the penalty function $s(t)$.
According to Section \ref{subsec:metrics_evolution}, $s(t)$ evolves in the state space:
\begin{equation}
A=\{b^{j}: j \geq 0, b^{j}=dp\frac{j(j+1)}{2}\}
\end{equation} 
\item \textbf{Actions}: The action at time $t$, denoted by $a(t)$, specify if the user is scheduled (value $1$) or not (value $0$).
\item \textbf{Transitions probabilities}: The transitions probabilities between the different states.
\item \textbf{Cost}: We let the instantaneous cost of the MDP, $C(s(t),a(t))$, be equal to $s(t)+Wa(t)$.
\end{itemize}
The optimal policy $\phi^*$ of the one-dimensional problem \eqref{eq:individual_dual_problem} can be obtained by solving the following Bellman equation for each state $b^j$:
\begin{align}
\theta + V(b^{j})=\min\big\{&b^{j}+V(b^{j+1});\nonumber \\
&b^{j}+W+\rho  V(b^{0})+(1-\rho)V(b^{j+1})\big\} 
\label{eq:bellman_general}
\end{align}
where $\theta$ is the optimal value of the problem, $V(b^j)$ is the differential cost-to-go function.
Instead of resolving the equation \eqref{eq:bellman_general}, we will limit ourselves to study the structure of the optimal scheduling policy of \eqref{eq:bellman_general}. To that end, we adopt the relative value iteration algorithm (\textbf{RVIA}) as follows:
\begin{align}
V_t(b^{j})=\min\big\{&b^{j}+V_t(b^{j+1});\nonumber \\
&b^{j}+W+\rho  V_t(b^{0})+(1-\rho)V_t(b^{j+1})\big\} 
\label{eq:bellman_equation_time_t}
\end{align}

\begin{theorem}\label{theo:threshold_policy}
The optimal solution of the problem in (\ref{eq:individual_dual_problem}) is an increasing threshold policy. Explicitly, there exists $n$ such that when the current state $b^{j} < b^{n}$, the prescribed action is a passive action, and when $ b^{j} \geq b^{n}$, the prescribed action is an active action.
\end{theorem}
\begin{IEEEproof}
See Appendix \ref{app:theo:threshold_policy}.
\end{IEEEproof}

\subsection{Indexability and Whittle's index expressions}
In order to establish the indexability of the problem and find the Whittle's index expressions, we provide the steady-state form of the problem in (\ref{eq:individual_dual_problem}) under a given threshold policy $n$. Explicitly:
\begin{equation}
\begin{aligned}
& \underset{n\in \mathbb{N}^*}{\text{minimize}} 
& & \overline{s^{n}}+W\overline{a^n}
\end{aligned}
\label{thresholdobjective}
\end{equation}
where $\overline{s^{n}}$ is the average value of the penalty function with respect to the AoII metric, and $\overline{a^n}$ is the average active time under threshold policy $n$. Specifically: 
\begin{align}
\overline{s^{n}}&=\lim_{T\to+\infty} \text{sup}\:\frac{1}{T}\mathbb{E}^{n}\Big(\sum_{t=0}^{T-1}s(t)|s(0),tp(n)\Big)\label{eq:average_age}\\
\overline{a^n}&=\lim_{T\to+\infty} \text{sup}\:\frac{1}{T}\mathbb{E}^{n}\Big(\sum_{t=0}^{T-1}a(t)|s(0),tp(n)\Big)\label{eq:average_active_time}
\end{align}
where $tp(n)$ denotes the threshold policy $n$.
With the aim of computing $\overline{s^{n}}$ and $\overline{a^n}$, we derive the stationary distribution of the Discrete Time Markov Chain, DTMC that represents the evolution of AoII under threshold policy $n$. Specifically: 
\begin{proposition}\label{prop:stationary_distribution}
For a given threshold $n$, the DTMC admits $u^n(b^{j})$ as its stationary distribution:
\begin{equation}
 u^n(b^{j})=\left\{
    \begin{array}{ll}
        \frac{\rho  }{n\rho  +1} & \text{if} \ 0 \leq j \leq n  \\
        (1-\rho  )^{j-n} \frac{\rho  }{n\rho  +1  } & \text{if} \ j \geq n+1\\
    \end{array}
\right.
\end{equation}
\label{stationarydistribution}
\end{proposition}
\vspace{-10pt}
\begin{IEEEproof}
The proof can be found in Appendix \ref{app:prop:stationary_distribution}.
\end{IEEEproof}

Leveraging the above results, we provide the closed-form expression of the AoII under any threshold policy.

\begin{proposition}\label{prop:mean_age_expression}
For a given threshold $n$, the average AoII under the threshold policy is $\overline{s^{n}}$:
\begin{align}
\overline{s^{n}}=&dp\frac{\rho }{n\rho  +1}[\frac{1}{6}n^3+\frac{1}{2\rho}n^2+\frac{6-\rho^2-3\rho}{6\rho^2}n+\frac{1-\rho}{\rho^3}]
\end{align}
\end{proposition}

\begin{IEEEproof}
See Appendix \ref{app:prop:mean_age_expression}.
\end{IEEEproof}

\begin{proposition}
For any given threshold $n$, the active average time is $\overline{a^n}$:
\begin{align}
\overline{a^n}=&\frac{1}{n\rho +1}
\end{align}
\end{proposition}
\begin{IEEEproof}
Exploiting the results in Proposition \ref{prop:stationary_distribution} and according to the expression \eqref{eq:average_active_time}, we have:
\begin{equation}
\overline{a^n}=\sum_{j=n}^{+\infty} u^n(j)=\frac{1}{n\rho +1}
\end{equation}
Hence, we obtain our desired results.
\end{IEEEproof}

To ensure the existence of the Whittle's indices, we need first to establish the indexability property for all users' classes.
A class is indexable if the set of states in which the passive action is the optimal action with respect to the optimal solution of Problem~\eqref{eq:individual_dual_problem}, increases with the Lagrangian parameter $W$. One can see \cite{kriouile2021minimizing} for a rigorous Definition of Indexability property as well as Whittle's index.  
We note that in the sequel, we precise the indices of users to differentiate between them.

\begin{proposition}
For each user $i$, the one-dimensional problem is indexable.
\end{proposition}
\begin{IEEEproof}
It is sufficient to show that $\overline{a_i^n}$ decreases with $n$ to establish our desired result (see \cite{maatouk2020optimality}). Indeed, we have that:
\begin{equation}
\overline{a_i^{n+1}}-\overline{a_i^n}=-\frac{\rho}{(n\rho+1)(n\rho+\rho+1)}\leq 0
\end{equation}
That concludes the proof.
\end{IEEEproof}
As the indexability property has been established in the above proposition, we can now assert the existence of the Whittle's index.
With the intention of comparing the two metrics AoI and AoII, we provide in the following Theorem the Whittle's index values distinguishing between two cases: the first case where we consider the AoI metric, and the second one where we consider the AoII metric.
For a sake of clarity, we recall that the state space of AoI for a given user $i$ is $A_i^{aoi}=\{c_i^j: c_i^j=j, j \in \mathbf{N} \}$    
\begin{theorem}\label{theo:Whittle_index_expressions}
For any user $i$, the Whittle's index is:
\begin{itemize}
\item AoI:
\begin{equation}
W_i(c_i^{n})=\frac{n(n+1)\rho_i}{2}+n+1
\end{equation}
\item AoII:
\begin{align}\label{eq:WIP_AoII_expressions}
W_i(b_i^{n})=&d_ip_i \biggl[\frac{1}{3}\rho_i n^3+(1+\frac{\rho_i}{2})n^2+(1+\frac{\rho_i}{6}+\frac{1}{\rho_i})n \nonumber \\ +&\frac{1}{\rho_i}\biggr]=d_i p_i w_i(n)
\end{align}
\end{itemize}
\end{theorem}
\begin{proof}
The proof can be found in Appendix \ref{app:theo:Whittle_index_expressions}.
\end{proof}
Based on the above proposition, we provide in the following the Whittle's index scheduling policy for the original problem \eqref{eq:original_problem}.
\begin{algorithm}
\caption{Whittle's index scheduling policy}\label{euclid}
\begin{algorithmic}[1]
\State At each time slot $t$, compute the Whittle's index of all users in the system using the expressions given in Proposition \ref{theo:Whittle_index_expressions}.
\State Allocating the $M$ channels to the $M$ users having the highest Whittle's index values at time $t+1$.
\end{algorithmic}
\end{algorithm}\\

\section{Numerical Results}\label{sec:num_reslt}
Our goal in this section is to compare the average empirical age of incorrect information under the developed Whittle's index policy WIP-AoII to the one under the baseline policy, denoted by WIP-AoI (derived in \cite{maatouk2020optimality}), that considers the standard AoI metric.
More precisely, we plot  $C^{\phi,N}=\frac{1}{N}\lim_{T\to+\infty} \text{sup}\:\frac{1}{T}\mathbb{E}^{\phi}\Big(\sum_{t=0}^{T-1}\sum_{i=1}^{N}s_i^{emp,\phi}(t)|\boldsymbol{s}^{emp}(0),\phi\Big)$ for $\phi$ equals to WIP-AoII and WIP-AoI, in function of $N$, where $s^{emp,\phi}_i(\cdot)$ evolves as follows:
\begin{itemize}
\item If $\phi_i(t+1)=1$: If the packet is successfully transmitted, then $\hat{X}_i(t+1)=X_i(t+1)$. Hence the AoII will move to the state $0$. 
If the packet is unsuccessfully transmitted, then the monitor maintains the last estimated value of the process of interest, i.e. $\hat{X}_i(t+1)=\hat{X}_i(t)$.
We have $s_i^{emp,\phi}(t)$ evolves in the state space $\{\sum_{k=0}^j k=\frac{(j+1)j}{2}: j \in \mathbb{N} \}$. 
Therefore, if $s_i^{emp,\phi}(t)=\frac{j(j+1)}{2}$, the value of $s_i^{emp,\phi}(t+1)$ is as follows: 
\begin{equation}
 s_i^{emp,\phi}(t+1)=\left\{
    \begin{array}{lll}
        d_i\frac{(j+2)(j+1)}{2} & w.p & (1-\rho_i)p_i\\
        d_i\frac{(j+1)j}{2} & w.p & (1-\rho_i)(1-p_i) \\
         0                  &w.p      & \rho_i \\
    \end{array}
\right.
\end{equation}
\item If $\phi_i(t)=0$: The monitor maintains the last estimated value of the process of interest, i.e. $\hat{X}_i(t+1)=\hat{X}_i(t)$.
Therefore, the value of $s_i^{emp,\phi}(t+1)$ is as follows:
\end{itemize}
\begin{equation}
 s_i^{emp,\phi}(t+1)=\left\{
    \begin{array}{lll}
        d_i\frac{(j+2)(j+1)}{2} & w.p & p_i\\
        d_i\frac{(j+1)j}{2} & w.p & 1-p_i \\
    \end{array}
\right.
\end{equation}

We showcase two scenarios of the network settings. In the first scenario, to shed light on the importance of taking into account the source parameters namely, $p_i$, in the derivation of Whittle's indices, we consider that the two classes share the same channel statistics, specifically $\rho_1=\rho_2$, while they don't have the same source parameters. In this case, we  compare the performance of WIP-AoII with WIP-AoI. For the second scenario, to highlight the importance of considering the Whittle index expressions derived through this paper precisely the function $w_i(.)$ in equation \eqref{eq:WIP_AoII_expressions}, we compare our proposed solution with the weighted-baseline policy denoted by WWIP-AoI where the expression of Whittle indices of WIP-AoI of class $i$ are multiplied by the factor $p_i d_i$. For the first scenario, we consider two classes with the respective parameters: i) Class 1:  $\rho_1=0.5$, $d_1=5$, $p_1=0.1$, and ii) Class 2: $\rho_2=0.5$, $d_2=5$, $p_2=0.9$. For the second scenario, we consider the following parameters:  i) Class 1:  $\rho_1=0.5$, $d_1=1$, $p_1=0.5$, and ii) Class 2: $\rho_2=0.5$, $d_2=100$, $p_2=0.5$.

\color{black}

\begin{figure}\label{fig:comp_aoi_maoii}
\centering
\includegraphics[scale=0.6]{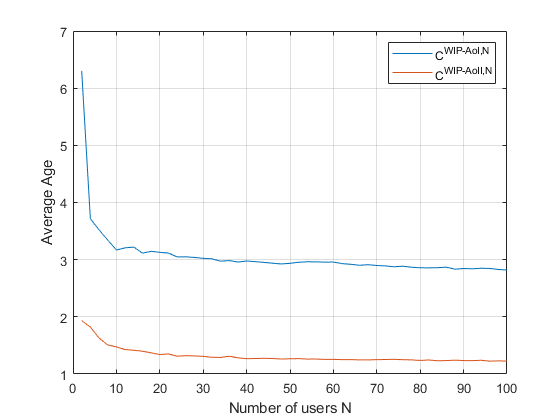}
\caption{Comparison between WIP-AoII and WIP-AoI in terms of the empirical AoII}
\end{figure}

\begin{figure}\label{fig:comp_waoi_maoii}
\centering
\includegraphics[scale=0.6]{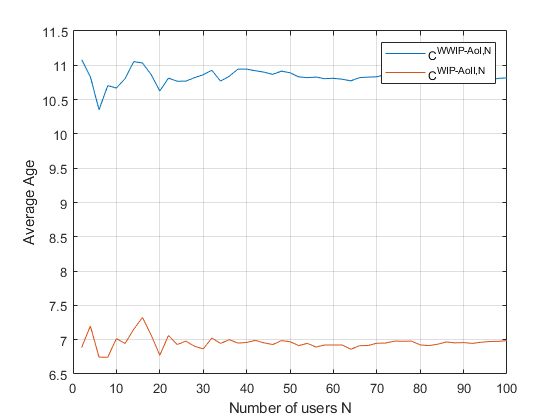}
\caption{Comparison between WWIP-AoII and WIP-AoI in terms of the empirical AoII}
\end{figure}

One can observe in Figures 2 and 3 that effectively WIP-AoII gives us better performance than WIP-AoI and WWIP-AoI in terms of minimizing the average empirical age of incorrect information considering the distances between the states.

\section{Conclusion}\label{sec:concl}
In this paper, we considered an instance of Age of Incorrect Information that takes into account the distances between the different states of a given Markov process. 
We considered a scheduling problem of a central entity that selects at each time slot a subset of the sensors to sample the Markovian sources and send instantaneously their updates in such a way to minimize the metric in question. Since the scheduler is unaware of the current state of the source, we computed the mean of AoII at each time slot. We then developed an efficient scheduling policy based on Whittle's index framework. Finally, we have provided numerical results that highlight the performance of our policy.

\bibliographystyle{IEEEtran} 
\bibliography{bibliography}

\begin{thebibliography}{10}
\providecommand{\url}[1]{#1}
\csname url@samestyle\endcsname
\providecommand{\newblock}{\relax}
\providecommand{\bibinfo}[2]{#2}
\providecommand{\BIBentrySTDinterwordspacing}{\spaceskip=0pt\relax}
\providecommand{\BIBentryALTinterwordstretchfactor}{4}
\providecommand{\BIBentryALTinterwordspacing}{\spaceskip=\fontdimen2\font plus
\BIBentryALTinterwordstretchfactor\fontdimen3\font minus
  \fontdimen4\font\relax}
\providecommand{\BIBforeignlanguage}[2]{{%
\expandafter\ifx\csname l@#1\endcsname\relax
\typeout{** WARNING: IEEEtran.bst: No hyphenation pattern has been}%
\typeout{** loaded for the language `#1'. Using the pattern for}%
\typeout{** the default language instead.}%
\else
\language=\csname l@#1\endcsname
\fi
#2}}
\providecommand{\BIBdecl}{\relax}
\BIBdecl

\bibitem{kaul2012real}
S.~Kaul, R.~Yates, and M.~Gruteser, ``Real-time status: How often should one
  update?'' in \emph{2012 Proceedings IEEE INFOCOM}.\hskip 1em plus 0.5em minus
  0.4em\relax IEEE, 2012, pp. 2731--2735.

\bibitem{maatouk2020optimality}
A.~Maatouk, S.~Kriouile, M.~Assaad, and A.~Ephremides, ``On the optimality of
  the whittle’s index policy for minimizing the age of information,''
  \emph{IEEE Transactions on Wireless Communications}, 2020.

\bibitem{hsu2019scheduling}
Y.-P. Hsu, E.~Modiano, and L.~Duan, ``Scheduling algorithms for minimizing age
  of information in wireless broadcast networks with random arrivals,''
  \emph{IEEE Transactions on Mobile Computing}, 2019.

\bibitem{kadota2018scheduling}
I.~Kadota, A.~Sinha, E.~Uysal-Biyikoglu, R.~Singh, and E.~Modiano, ``Scheduling
  policies for minimizing age of information in broadcast wireless networks,''
  \emph{IEEE/ACM Transactions on Networking}, vol.~26, no.~6, pp. 2637--2650,
  2018.

\bibitem{bedewy2019age}
A.~M. Bedewy, Y.~Sun, and N.~B. Shroff, ``The age of information in multihop
  networks,'' \emph{IEEE/ACM Transactions on Networking}, vol.~27, no.~3, pp.
  1248--1257, 2019.

\bibitem{maatouk2020status}
A.~Maatouk, Y.~Sun, A.~Ephremides, and M.~Assaad, ``Status updates with
  priorities: Lexicographic optimality,'' in \emph{2020 18th International
  Symposium on Modeling and Optimization in Mobile, Ad Hoc, and Wireless
  Networks (WiOPT)}.\hskip 1em plus 0.5em minus 0.4em\relax IEEE, 2020, pp.
  1--8.

\bibitem{sun2018age}
Y.~Sun, E.~Uysal-Biyikoglu, and S.~Kompella, ``Age-optimal updates of multiple
  information flows,'' in \emph{IEEE INFOCOM 2018-IEEE Conference on Computer
  Communications Workshops (INFOCOM WKSHPS)}.\hskip 1em plus 0.5em minus
  0.4em\relax IEEE, 2018, pp. 136--141.

\bibitem{SunElif2017}
Y.Sun, E.Uysal-Biyikoglu, R.D.Yates, C.E.Koksal, and N.B.Shroff, ``Update or
  wait: How to keep your data fresh,'' \emph{IEEE Transactions on Information
  Theory}, vol.~63, no.~11, pp. 7492--7508, 2017.

\bibitem{Oguz2022}
T.~K. Oguz, E.~T. Ceran, E.~Uysal, and T.~Girici, ``Implementation and
  evaluation of age-aware downlink scheduling policies in push-based and
  pull-based communication,'' \emph{IEEE Transactions on Communications},
  vol.~24, no.~5, p. 673, 2022.

\bibitem{zhong2018}
J.~Zhong, R.~D. Yates, and E.~Soljanin, ``Two freshness metrics for local cache
  refresh,'' in \emph{IEEE International Symposium on Information Theory
  (ISIT)}.\hskip 1em plus 0.5em minus 0.4em\relax IEEE, 2018, pp. 1924--1928.

\bibitem{Ayan2019}
O.~Ayan, M.~Vilgelm, M.~Klugel, S.~Hirche, and W.~Kellerer,
  ``Age-of-information vs. value-of-information scheduling for cellular
  networked control systems,'' in \emph{in Proceedings of the 10th ACM/IEEE
  ICCPS}, 2019.

\bibitem{Escroc2019}
G.~Stamatakis and A.~T. N.~Pappas, ``Control of status updates for energy
  harvesting devices that monitor processes with alarms,'' in \emph{In Proc. of
  IEEE Globecom Workshops (GC Wkshps)}, 2019.

\bibitem{Chiariotti2022}
F.~C. et~al., ``Query age of information: Freshness in pull-based
  communication,'' \emph{IEEE Transactions on Communications}, vol.~70, no.~3,
  pp. 1606--1622, 2022.

\bibitem{ElifSemantic}
E.~Uysal and et~al, ``Semantic communications in networked systems: A data
  significance perspective,'' \emph{IEEE Network}, vol.~36, no.~4, pp.
  233--240, 2022.

\bibitem{maatouk2020age}
A.~Maatouk, S.~Kriouile, M.~Assaad, and A.~Ephremides, ``The age of incorrect
  information: A new performance metric for status updates,'' \emph{IEEE/ACM
  Transactions on Networking}, vol.~28, no.~5, pp. 2215--2228, 2020.

\bibitem{Ayik2023}
M.~Ayik, E.~T. Ceran, and E.~Uysal, ``Optimization of aoii and qaoii in
  multi-user links,'' in \emph{available on arviv, arXiv:2305.00191}, 2023.

\bibitem{maatouk2020agenew}
A.~Maatouk, M.~Assaad, and A.~Ephremides, ``The age of incorrect information:
  an enabler of semantics-empowered communication,'' \emph{arXiv preprint
  arXiv:2012.13214}, 2020.

\bibitem{kam2020age}
C.~Kam, S.~Kompella, and A.~Ephremides, ``Age of incorrect information for
  remote estimation of a binary markov source,'' in \emph{IEEE INFOCOM
  2020-IEEE Conference on Computer Communications Workshops (INFOCOM
  WKSHPS)}.\hskip 1em plus 0.5em minus 0.4em\relax IEEE, 2020, pp. 1--6.

\bibitem{kriouile2021minimizing}
S.~Kriouile and M.~Assaad, ``Minimizing the age of incorrect information for
  real-time tracking of markov remote sources,'' in \emph{2021 IEEE
  International Symposium on Information Theory (ISIT)}, 2021, pp. 2978--2983.

\bibitem{chen2021minimizing}
Y.~Chen and A.~Ephremides, ``Minimizing age of incorrect information for
  unreliable channel with power constraint,'' \emph{arXiv preprint
  arXiv:2101.08908}, 2021.

\bibitem{ChenAoII2021}
------, ``Scheduling to minimize age of incorrect information with imperfect
  channel state information,'' \emph{Entropy}, vol.~23, no.~12, 2021.

\bibitem{Maialen18}
M.~Larranaga, M.~Assaad, A.~Destounis, and G.~Paschos, ``Asymptotically optimal
  pilot allocation over markovian fading channels,'' \emph{IEEE Transactions on
  Information Theory}, vol.~64, no.~7, pp. 5395--5418, 2018.

\bibitem{tripathi2019whittle}
V.~Tripathi and E.~Modiano, ``A whittle index approach to minimizing functions
  of age of information,'' in \emph{2019 57th Annual Allerton Conference on
  Communication, Control, and Computing (Allerton)}.\hskip 1em plus 0.5em minus
  0.4em\relax IEEE, 2019, pp. 1160--1167.

\bibitem{sombabu2020whittle}
B.~Sombabu, A.~Mate, D.~Manjunath, and S.~Moharir, ``Whittle index for
  aoi-aware scheduling,'' in \emph{2020 International Conference on
  COMmunication Systems \& NETworkS (COMSNETS)}.\hskip 1em plus 0.5em minus
  0.4em\relax IEEE, 2020, pp. 630--633.

\bibitem{SaadAoI2022}
S.~Kriouile, M.~Assaad, and A.~Maatouk, ``On the global optimality of whittle's
  index policy for minimizing the age of information,'' \emph{IEEE Transactions
  on Information Theory}, vol.~68, no.~1, pp. 572--600, 2022.

\bibitem{papadimitriou1999complexity}
C.~H. Papadimitriou and J.~N. Tsitsiklis, ``The complexity of optimal queuing
  network control,'' \emph{Mathematics of Operations Research}, vol.~24, no.~2,
  pp. 293--305, 1999.

\bibitem{weber1990index}
R.~R. Weber and G.~Weiss, ``On an index policy for restless bandits,''
  \emph{Journal of Applied Probability}, vol.~27, no.~3, pp. 637--648, 1990.

\bibitem{kriouile:hal-03437753}
\BIBentryALTinterwordspacing
S.~Kriouile, M.~Assaad, and M.~Larranaga, ``{Asymptotically Optimal Delay-aware
  Scheduling in Queueing Systems},'' \emph{{Journal of Communications and
  Networks}}, 2021. [Online]. Available:
  \url{https://hal.archives-ouvertes.fr/hal-03437753}
\BIBentrySTDinterwordspacing

\end{thebibliography}

\begin{appendices}
\section{Proof of Lemma \ref{lem:random_variable}}\label{app:lem:random_variable}
We have for $g_i(t) \leq u \leq t$, $A_i(u)=X_i(u)-\hat{X}_i(u)$. By definition of $g_i(t)$, $\hat{X}_i(u)=X_i(g_i(t))$. Then $\hat{X}_i(u)$ is a fixed constant. Whereas $X_i(t)$ is unknown by the monitor. Accordingly, it is viewed as a random variable by the monitor. Baring in mind the system dynamics showcased in Section \ref{sec:Syst_mod}, $X_i(u)-\hat{X}_i(u)$ follows a binomial distribution with parameters $u-g_i(t)$ and $p_i$. Hence, the probability that $A_i(u)=u-g_i(t)-k$ is $(p_i)^{u-g_i(t)-k}.\binom{u-g_i(t)}{k}(1-p_i)^k$. That concludes the proof.

\section{Proof of Lemma \ref{lem:CME_expression}}\label{app:lem:CME_expression}
\begin{align}
&n_i(t)=\sum_{u=g_i(t)}^t d_i\mathbb{E}[A_i(u)] \nonumber \\
&=\sum_{u=g_i(t)}^t \sum_{k=0}^{u-g_i(t)} d_i. (u-g_i(t)-k) (p_i)^{u-g_i(t)-k} \nonumber\\
& \ \ \ \ \ \ \ \ \ \times \binom{u-g_i(t)}{k}(1-p_i)^k \nonumber \\
&=\sum_{u=g_i(t)}^t \sum_{k=0}^{u-g_i(t)} d_i k (p_i)^{k}.\binom{u-g_i(t)}{k}(1-p_i)^{u-g_i(t)-k} \nonumber \\
\end{align}
As we have, for all integer $n$ and $0\leq p \leq 1$, $\sum_{k=0}^n k.p^k.
(1-p)^{n-k}.\binom{n}{k}=np$, then:
\begin{align}
n_i(t)&=\sum_{u=g_i(t)}^t d_i p_i(u-g_i(t)) \nonumber \\
&=d_i p_i\frac{(t-g_i(t)+1)(t-g_i(t))}{2} 
\end{align}

\section{Proof of theorem \ref{theo:threshold_policy}}\label{app:theo:threshold_policy}

We provide first an useful lemma. 

\begin{lemma}\label{lem:MAoII_increasing}
$b^{j}$ is increasing with $j$
\end{lemma} 

\begin{IEEEproof}
From the expression of $b^j$, it is clear that this later is increasing with $j$. 
\end{IEEEproof}

Based on this lemma, we prove the following lemma.
\begin{lemma}\label{lem:V_increasing}
$V(.)$ is increasing with $b^j$.
\end{lemma}
\begin{IEEEproof}
We prove the present lemma by induction using the Relative value iteration equation \eqref{eq:bellman_equation_time_t}. In fact, we show that $V_t(\cdot)$ is increasing and we conclude for $V(\cdot)$.\\
As $V_0(.)=0$, then the property holds for $t=0$.
If $V_t(.)$ is increasing with $b$, we show that for $b^{j} \leq b^{i}$, $V_{t+1}^0(b^{j}) \leq V_{t+1}^0(b^{i})$ and $V_{t+1}^1(b^{j}) \leq V_{t+1}^1(b^{i})$ where for each $k \in \mathbf{N}$:
\begin{align}
V_{t+1}^0(b^{k})&=b^{k}+V_t(b^{k+1}) \\
V_{t+1}^1(b^{k})&=b^{k}+W+\rho V_t(b^{0})+(1-\rho)V_t(b^{k+1}) 
\end{align}
We have that:
\begin{equation}
V_{t+1}^0(b^{j}) - V_{t+1}^0(b^{i})=b^{j}-b^{i}+(V_{t}(b^{j+1}) - V_{t}(b^{i+1}))
\end{equation}
According to Lemma \ref{lem:MAoII_increasing}, given that $b^{j} \leq b^{i}$, then $j \leq i$. That means $b^{j+1} \leq b^{i+1}$. Therefore, since $V_t(.)$ is increasing with $b^j$, we have that:
$V_{t+1}^0(b^{j}) - V_{t+1}^0(b^{i}) \leq 0$. \\
As consequence, $V_{t+1}^0(\cdot)$ is increasing with $b^{j}$. 

In the same way, we have:
$$V_{t+1}^1(b^{j}) - V_{t+1}^1(b^{i})=b^{j}-b^{i}+(1-\rho)(V_{t}(b^{j+1}) - V_{t}(b^{i+1}))$$
Hence: 
\begin{equation}
V_{t+1}^1(b^{j}) - V_{t+1}^1(b^{i}) \leq 0
\end{equation}
As consequence, $V_{t+1}^1(\cdot)$ is increasing with $b^{j}$.\\
Since $V_{t+1}(.)=\min\{V^0_{t+1}(\cdot),V^1_{t+1}(\cdot)\}$, then $V_{t+1}(.)$ is increasing with $b^j$. Accordingly, we demonstrate by induction that $V_t(.)$ is increasing for all $t$. Knowing that $\underset{t \rightarrow +\infty}{\text{lim}} V_t(b^j)=V(b^j)$, $V(.)$ must be also increasing with $b^j$.
\end{IEEEproof}

We define:
\begin{equation}
\Delta V(b^{j})=V^1(b^j)-V^0(b^j)
\end{equation}
where $\underset{t \rightarrow +\infty}{\text{lim}} V_t^0(b^j)=V^0(b^j)$ and $\underset{t \rightarrow +\infty}{\text{lim}} V_t^1(b^j)=V^1(b^j)$.\\
Subsequently, $\Delta V(b^{j})$ equals to:
\begin{equation}
\Delta V(b^{j})=\rho [\frac{W}{\rho }+V(b^0)-V(b^{j+1})]
\end{equation}
According to Lemma \ref{lem:V_increasing}, $V(.)$ is increasing with $b^{j+1}$. Therefore, $\Delta V(b^{j})$ is decreasing with $b^j$. Hence, there exists $b^{n}$ such that for all $b^j \leq b^{n}$, $\Delta V(b^{j})\geq 0$, and for all $b^j > b^{n}$, $\Delta V(b^{j}) <0$. Given that the optimal action for state $b^j$ is the one that minimizes $\min\{V^0(\cdot),V^1(\cdot)\}$, then for all $b^j \leq b^{n}$, the optimal decision is to stay idle since $\min\{V^0(b^j),V^1(b^j)\}=V^0(b^j)$, and for all $b^j > b^{n}$, the optimal decision is to transmit since $\min\{V^0(b^j),V^1(b^j)\}=V^1(b^j)$. Specifically, as $b^{j}$ is increasing with $j$, there exists $n$ such that for all $j < n$, the optimal action is passive action, and for all $j \geq n$, the optimal action is the active one. 

\section{Proof of Proposition \ref{prop:stationary_distribution}}\label{app:prop:stationary_distribution}
In order to demonstrate this proposition, we need to resolve the full balance equation under threshold policy $n$ at each state $b^{j}$:
\begin{equation}
u^n(b^{j})=\sum_{i=0}^{+\infty} pt^n(i \rightarrow j)u^n(b^{i})
\end{equation}
where $pt^n(i \rightarrow j)$ denotes the transitioning probability from the state $b^{i}$ to the state $b^{j}$ under threshold policy $n$. 
After some computations, we obtain the desired result.  

\section{Proof of Proposition \ref{prop:mean_age_expression}}\label{app:prop:mean_age_expression}
Exploiting the results of Proposition \ref{prop:stationary_distribution} and by definition of $\overline{s^{n}}$ given in \eqref{eq:average_age}, we have that: $\overline{s^{n}}=\sum_{j=0}^{+\infty} b^j u^n(b^j)$.

Therefore, using the expression of $b^{j}$ for $j \geq 0$, we have that:
\begin{align}
\overline{s^{n}}&=dp\frac{\rho}{n\rho+1} [\sum_{j=0}^n \frac{j(j+1)}{2} +\sum_{j=n+1}^{+\infty} \frac{j(j+1)}{2} (1-\rho)^{j-n}]\nonumber \\
&=dp \frac{\rho}{n\rho+1}[\sum_{j=0}^n \frac{j^2}{2}+\sum_{j=0}^n \frac{j}{2}+\sum_{j=n+1}^{+\infty} j(j+1)\frac{(1-\rho)^{j-n}}{2}]
\end{align}\\

We have also: \\ $\sum_{j=0}^n \frac{j^2}{2}=\frac{1}{2}(\frac{n^3}{3}+\frac{n^2}{2}+\frac{n}{6})$ and $\sum_{j=0}^n \frac{j}{2}=\frac{1}{2}\frac{n(n+1)}{2}$. 
As for the last term, we have that:
\begin{align}
&\sum_{j=n+1}^{+\infty} j(j+1)\frac{(1-\rho)^{j-n}}{2}
=\sum_{j=1}^{+\infty} (j+n)(j+1+n)\frac{(1-\rho)^{j}}{2}\nonumber \\
&=\sum_{j=1}^{+\infty} j(j+1)\frac{(1-\rho)^{j}}{2}+\frac{n^2}{2} \frac{1-\rho}{\rho}+\frac{n}{2}\frac{1-\rho^2}{\rho^2} \nonumber \\ &+\frac{n}{2}(1-\rho)\frac{1}{\rho^2}
\end{align}
We have that:
\begin{align}
\sum_{j=1}^{+\infty} j(j+1)(1-\rho)^{j}&=\sum_{j=0}^{+\infty} j(j+1)(1-\rho)^{j}\nonumber \\
&=(1-\rho)\sum_{j=0}^{+\infty}\frac{\partial^2(1-\rho)^{j+1}}{\partial (1-\rho)^2}
\end{align}
Leveraging that, and given that $\sum_{j=0}^{+\infty}(1-\rho)^{j+1}=\frac{1-\rho}{\rho}$, then by deriving twice this term with respect to $1-\rho$, we get $\frac{1}{2}\sum_{j=1}^{+\infty} j(j+1)(1-\rho)^{j}=\frac{1-\rho}{\rho^3}$.
Adding all terms together, we get:
\begin{align}
\overline{s^{n}}=&dp\frac{\rho }{n\rho  +1}[\frac{1}{6}n^3+\frac{1}{2\rho}n^2+\frac{6-\rho^2-3\rho}{6\rho^2}n+\frac{1-\rho}{\rho^3}]
\end{align}
As consequence, we get our the desired results.
\section{Proof of Theorem \ref{theo:Whittle_index_expressions}}\label{app:theo:Whittle_index_expressions}
The investigations regarding the expression of the Whittle's index for Age of Information metric have been already done in \cite{maatouk2020optimality}. 
To that extent, in this proof, we do the analysis only for the adapted AoII metric considered throughout our paper.
We first define the sequence $W_i(b_i^{n})$
as the intersection points between $\overline{b_i^n}+W\overline{a_i^n}$ and $\overline{b_i^{n+1}}+W\overline{a_i^{n+1}}$. Explicitly:
\begin{equation}
W_i(b_i^{n})=\frac{\overline{b_i^{n+1}}-\overline{b_i^{n}}}{\overline{a_i^n}-\overline{a_i^{n+1}}}
\end{equation}
According to the results in [32, Corollary 2.1], if $W_i(b_i^{n})$ 
is increasing with $b_i^{n}$, then the Whittle's
index for any state $b_i^{n}$ is nothing but $W_i(b_i^{n})$. To that extent, we prove that $W_i(b_i^{n})$ is increasing with $b_i^{n}$. However, since $b_i^{n}$ is increasing with $n$, it is sufficient to show that $W_i(b_i^{n})$ is increasing with $n$ to establish the desired result.  

Therefore, we first seek a closed-form expression of the intersection point $W_i(b_i^{n})$, we obtain:
\begin{align}
W_i(b_i^{n})=d_ip_i [\frac{1}{3}\rho_i n^3+(1+\frac{\rho_i}{2})n^2
+(1+\frac{\rho_i}{6}+\frac{1}{\rho_i})n+\frac{1}{\rho_i}] 
\end{align}
From the expression above, it is clear that $W_i(b_i^n)$ is increasing with $b_i^n$. 
Therefore $W_i(b_i^n)$ is the Whittle's index of the state $b_i^n$.
That concludes the proof. 

\end{appendices}

\end{document}